\documentclass[11pt,a4paper]{article}
\usepackage{fullpage}
\usepackage[utf8]{inputenc}
\usepackage[british]{babel}

\usepackage{amssymb,amsmath,array}
\usepackage{amsthm}
\usepackage{hyperref}
\usepackage{xspace}
\usepackage{amsmath,amsthm,amssymb}
\usepackage{graphicx} 
\usepackage{xcolor}
\usepackage{wrapfig}
\usepackage{subfigure}
\usepackage[font=scriptsize]{caption}

\DeclareMathOperator*{\sign}{sign}
\newcommand{\scp}[3][]{#1\langle #2, #3 #1\rangle}
\newcommand{\rop}{\cl A^{\rm v}}
\newcommand{\ropmtx}{\cl A}
\newcommand{\drop}{\cl B}
\newcommand{\ie}{\emph{i.e.}, }
\newcommand{\eg}{\emph{e.g.}, }

\newcommand{\bs}{\boldsymbol}
\newcommand{\bb}{\mathbb}
\newcommand{\cl}{\mathcal}
\newcommand{\ts}{\textstyle}

\newcommand{\iid}{%
  \ifmmode
  \mathrm{i.i.d.}%
  \else%
  i.i.d.\@\xspace%
  \fi%
}

\newcommand{\ignore}[1]{}
\date{\today}

\newtheorem{thm}{Theorem}
\newtheorem{proposition}[thm]{Proposition}
\newtheorem*{nnproposition}{Proposition}
\newtheorem{corollary}[thm]{Corollary}

\begin{document}
\title{ROP inception: signal estimation\\ with quadratic random sketching}

\author{Rémi Delogne, Vincent Schellekens, and Laurent Jacques%
\thanks{LJ is funded by the Belgian FNRS. Part of this work is funded by the FNRS under Grant n$^\circ$ T.0136.20 (Learn2Sense).}
\vspace{.3cm}\\
ISPGroup, INMA, ICTEAM, UCLouvain, Belgium.
}

\maketitle

\begin{abstract}
Rank-one projections (ROP) of matrices and quadratic random sketching of signals support several data processing and machine learning methods, as well as recent imaging applications, such as phase retrieval or optical processing units. In this paper, we demonstrate how signal estimation can be operated directly through such quadratic sketches---equivalent to the ROPs of the ``lifted signal'' obtained as its outer product with itself---without explicitly reconstructing that signal. Our analysis relies on showing that, up to a minor debiasing trick, the ROP measurement operator satisfies a generalised sign product embedding (SPE) property. 
In a nutshell, the SPE shows that the scalar product of a signal sketch with the \emph{sign} of the sketch of a given pattern approximates the square of the projection of that signal on this pattern. This thus amounts to an insertion (an \emph{inception}) of a ROP model inside a ROP sketch. The effectiveness of our approach is evaluated in several synthetic experiments.
\end{abstract}

\section{Introduction}
\label{sec:intro}

More and more algorithms in signal processing, optimisation, matrix algebra and machine learning rely on \emph{random projections}. They are used to relax the computational burden of ever-growing data flows with data-agnostic dimensionality reduction procedures, such as linear random projections for data embedding or compressive sensing~\cite{JL84}, while preserving specific information. Random projections are also used (in combination with nonlinear maps) to unfold (or embed) a dataset in larger dimensional feature spaces~\cite{saade2016random}. For example, this expanded domain can be more amenable to data separability or to clustering than the initial space, or endowed with specific kernels induced by the projection~\cite{kitchen}.

The \emph{estimation} of specific properties (or functions) of signals is crucial to many data processing techniques, hence the ability to perform it in the feature (``projected'') domain is of strong interest. For instance, one may be interested in processing video streams (\eg for traffic monitoring, industrial quality control or video surveillance) to deduce localised data characteristics, \ie restricted to a given area of the original field-of-view (\eg for traffic density estimation or object detection~\cite{Ka03}). Data processing tools based on random projection should thus comply with this goal, which can be challenging if the area-of-interest is unknown \emph{a priori}, or subject to change. 

In this work, we tackle the question of performing \emph{signal estimation}---here restricted to the estimation of some linear function of a signal---from a projected data stream provided by \emph{quadratic random projections} of signals (see Sec.~\ref{sec:hiddenROP}). This specific data-agnostic projection is related to the sketching\footnote{In this work, the term ``\emph{sketch}'' designates a generic data transformation (feature map) without restriction to the special case of dimensionality reduction; the dimension of a signal sketch can thus be larger that the input signal space dimension.} technique of rank-one projections (ROP)~\cite{chen,ROP}. Our work is motivated by a recent optical machine which performs  ultra-rapid and low-power computations of such quadratic sketches~\cite{saade2016random}. However, with this technology, these computations are performed in a black-box manner, which prohibits us to \emph{explicitly} access the random construction supporting the ROP model inside the localisation procedure---they are thus \emph{hidden} to us (\ie we do not have access to the adjoint ROP operator).      
Our main contribution amounts to showing theoretically that, up to some controlled distortion, signal estimation can be operated directly on quadratic signal sketches without reconstructing the signal (and hence avoiding often costly reconstruction methods~\cite{phaselift,ROP,flavours,chen}). More specifically, the square of this comparison---which is a ROP itself---can be approximated by projecting the sketched signal on the \emph{sign} of the sketched pattern, \ie summarising this last sketch to only the \emph{sign} its components (see Sec.~\ref{sec:sigestinDROP}). We thus achieves a sort of \emph{inception} of a ROP inside a ROP model. This result is thus similar to the techniques pursued in~\cite{SPWCM}, where specific signal processing tasks are shown to be computable directly from compressive measurements. Our theoretical analysis is achieved from a generalisation of the sign product embedding (SPE) property, initially developed in one-bit compressive sensing~\cite{JKD13,flavours}.     

To demonstrate the efficacy of our approach, we consider several synthetic signal estimation scenarios in Sec.~\ref{sec:exper}. A first experiment quantifies the approximation error induced by our approach by considering a pattern that is orthogonal to the signal of interest. Next, we show how we can localise a rotating disk in one of the four quadrants of an image using only the image ROPs. Finally, we propose to classify MNIST handwritten digit images~\cite{mnist} directly in the DROP domain, with accuracy comparable to a direct processing of these images.  

\section{Sketching with hidden rank one projections}
\label{sec:hiddenROP}

In this work, we consider the \emph{quadratic data sketching} mechanism, which consists in taking a series of $m$ quadratic measurements $(\bs a_i^\top\bs x)^2$ of a signal of interest $\bs x\in\bb R^n$, with a set of $m$ random vectors $\{\bs a_i\}_{i=1}^m \subset \bb R^n$. The sketching operator $\rop$ applied to the vector $\bs x$ is thus defined as
\begin{equation}
\label{eq:originalROP}
\rop: \bs x \in \bb R^n \mapsto \rop(\bs x) := \big( (\bs a_i^\top\bs x)^2 \big)_{i=1}^m := \big( (\bs a_1^\top\bs x)^2, \ldots, (\bs a_m^\top\bs x)^2 \big) \in \bb R^m_+.
\end{equation}
We assume that, while the operator $\rop$ can be computed, we cannot explicitly access to the random vectors $\{\bs a_i\}_{i=1}^m$. This restriction is indeed required by a recent optical technology, named Optical Processing Unit (OPU)~\cite{saade2016random}. An OPU allows us to compute all the components of $\rop(\bs x)$ in a reproducible way using the physical properties of multiple scattering of coherent light in random media, which is thus extremely fast and power-efficient (even if $m \simeq n$). In this context, the vectors $\{\bs a_i\}_{i=1}^m$ are fixed, but \emph{hidden} to us. Moreover, following the observations made in~\cite{saade2016random}, we assume that each random vector $\bs a_i$ is \iid as a Gaussian random vector\footnote{This optical projection is actually modelled by quadratic projections over complex random vectors, but we here work in the real field for the sake of simplicity.} $\bs a \sim \cl N(\bs 0, \bs I_n)$, with identity covariance $\bs I_n$.

As observed in the context of phase retrieval~\cite{phaselift}, the operator $\rop$ amounts to a rank-one projection of the \emph{lifted signal}, \ie the rank-one matrix $\bs X = \bs x \bs x^\top \in \bb R^{n \times n}$, onto the rank-one random matrices $\{\bs A_i := \bs a_i \bs a_i^\top\}_{i=1}^m \subset \bb R^{n \times n}$, as defined by the equivalence 
$$
\rop(\bs x)=\big( (\bs a_i^\top\bs x)^2 \big)_{i=1}^m=\big( \bs a_i^\top\bs x\bs x^\top\bs a_i \big)_{i=1}^m=\big( \langle\bs A_i,\bs X\rangle \big)_{i=1}^m =: \ropmtx(\bs X),
$$
where $\langle \cdot, \cdot \rangle$ is the Frobenius inner product~\cite{chen,ROP}. We thus use ``quadratic sketch'' and ``ROP measurements'' interchangeably.   

As the ROP operator is biased---\ie given a signal $\bs x$, there is no constant $c>0$ such that $\frac{c}{m} \bb E \|\rop(\bs x)\|^2$ equals $\|\bs x\|^4 = \|\bs X\|_F^2$---it is useful to introduce the \textit{debiased} ROP operator (DROP)~\cite{chen} 
\begin{equation}
\label{eq:DROP}
\drop: \bs x \in \bb R^{n} \mapsto \drop(\bs x)= \big( \rop_{2i}(\bs x) - \rop_{2i+1}(\bs x) \big)_{i=1}^{m}.
\end{equation}
An OPU can trivially implement this operation by applying finite pixel differences in its focal plane. It can be shown that $\ts\frac{1}{4m} \bb E \|\drop(\bs x)\|^2 = \|\bs x\|^4$~\cite[Lemma 4]{chen}. While the approximation $\ts\frac{1}{4m} \|\drop(\bs x)\|^2 \approx \|\bs x\|^4$ is hard to achieve at reasonable values of $m$---$\drop$ respects the restricted isometry property (RIP) in only a few restrictive settings \cite{chen,ROP}---we leverage below another useful property of $\drop$.

\section{Signal estimation in the DROP domain}
\label{sec:sigestinDROP}

Our objective is to show that we can do (approximate) signal estimation---when this estimation amounts to estimating a linear function of the signal---directly from random quadratic signal sketches, and thus without knowing the observed signal. This is possible by demonstrating that the DROP operator respects, with high probability, the (local) sign product embedding (SPE) property~\cite{JKD13}, whose proof is postponed to App.~\ref{sec:proofs}. This SPE is here instantiated on a signal space $\cl S$ consisting of $k$ sparse signals of $\Sigma_k := \{\bs v \in \bb R^n: |{\rm supp}(\bs v)|\leq k\}$. However, by rotational symmetry of the Gaussian distribution, the proposition below is also valid for sets of $k$ sparse signals in an orthonormal basis $\bs \Psi \in \bb R^{n\times n}$ (such as the wavelet or Fourier bases).
\begin{proposition}
\label{prop:drop-spe}
Given a fixed unit vector $\bs u \in \bb R^n$, $\kappa = \pi/4$, and a distortion $0<\delta <1$, provided that 
\begin{equation}
\label{eq:sample-complex-drop-spe}
\ts \ts m \geq C \delta^{-2} k \log(\frac{n}{k\delta}),    
\end{equation}
then, with probability exceeding $1 - C \exp(-c \delta ^2 m)$, for all $k$-sparse signals $\bs x \in \Sigma_k := \{\bs v \in \bb R^n: |{\rm supp}(\bs v)|\leq k\}$, $\cl B$ respects the SPE over $\Sigma_k$, \ie 
\begin{equation}
\label{eq:drop-spe}
\ts \Big|\frac{\kappa}{m} \scp{\sign(\drop(\bs u))}{\drop(\bs x)} - \scp{\bs u}{\bs x}^2 \Big| \leq \delta \|\bs x\|^2,    
\end{equation}    
with $\sign$ the sign operator applied componentwise on vectors. 
\end{proposition}
This proposition states that, provided that $m$ is large compared to the dimension of the signal space $\cl S$ (here $\cl S = \Sigma_k$ and we thus need $m=\Omega(\delta^{-2} k)$ up to log factors), for any vector $\bs x \in \cl S$, projecting its sketch $\drop(\bs x)$ on $\sign(\cl B(\bs u)) \in \{\pm 1\}^m$ is a proxy for $\scp{\bs u}{\bs x}^2$---a ROP of $\bs x\bs x^\top$ by $\bs u \bs u^\top$. The spirit of this result is thus similar to the approach of~\cite{SPWCM} in linear compressive sensing~\cite{SPWCM}, where a set of signal processing techniques (including signal estimation) are proved to be applicable in the compressed signal measurements. Let us emphasise that the error of the approximation~\eqref{eq:drop-spe} is bounded by $\delta \|\bs x\|^2$, that is $O(\sqrt{k/m}\,\|\bs x\|^2)$ (up to log factors) by saturating~\eqref{eq:sample-complex-drop-spe}. Therefore, the equation~\eqref{eq:drop-spe} is useful only if $\scp{\bs u}{\bs x/\|\bs x\|}^2$ is sufficiently large compared to $\delta$. Incidentally, this allows us to estimate $\delta$ as a function of $m$ by taking a vector $\bs u$ orthogonal to $\bs x$ (see Sec.~\ref{sec:exper}). Notice that, by a simple union bound argument, Eq.~\ref{eq:drop-spe} in Prop.~\ref{prop:drop-spe} can be shown to hold with the same probability bound for all $\bs u$ in a finite set of $S$ unit vectors provided $m \geq C \delta^{-2} (k \log(\frac{n}{k\delta}) + \log S)$ (see Cor.~\ref{cor:drop-spe-set} in App.~\ref{sec:proofs}).

\section{Experiments}
\label{sec:exper}

\begin{figure}[t]
\centering
\includegraphics[width=0.45\textwidth]{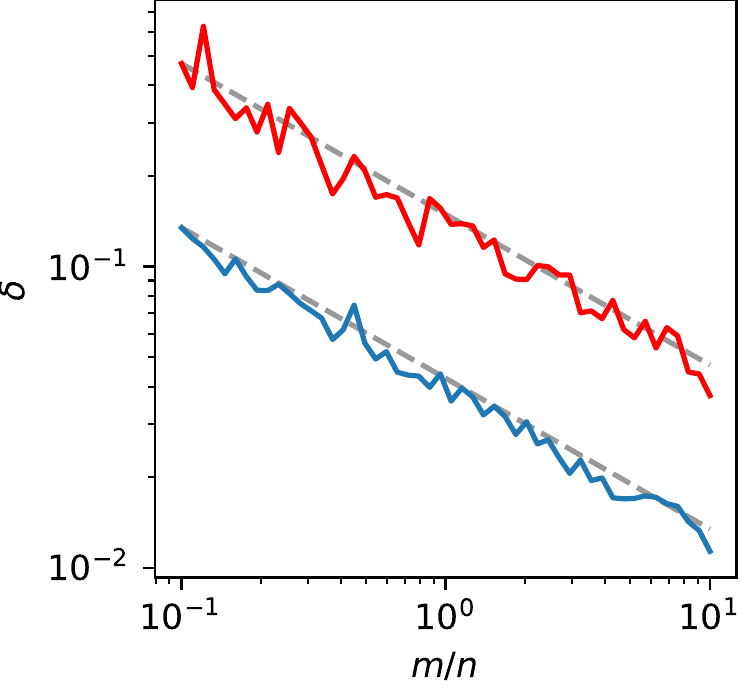}\qquad 
\raisebox{10mm}{\includegraphics[width=0.35\textwidth]{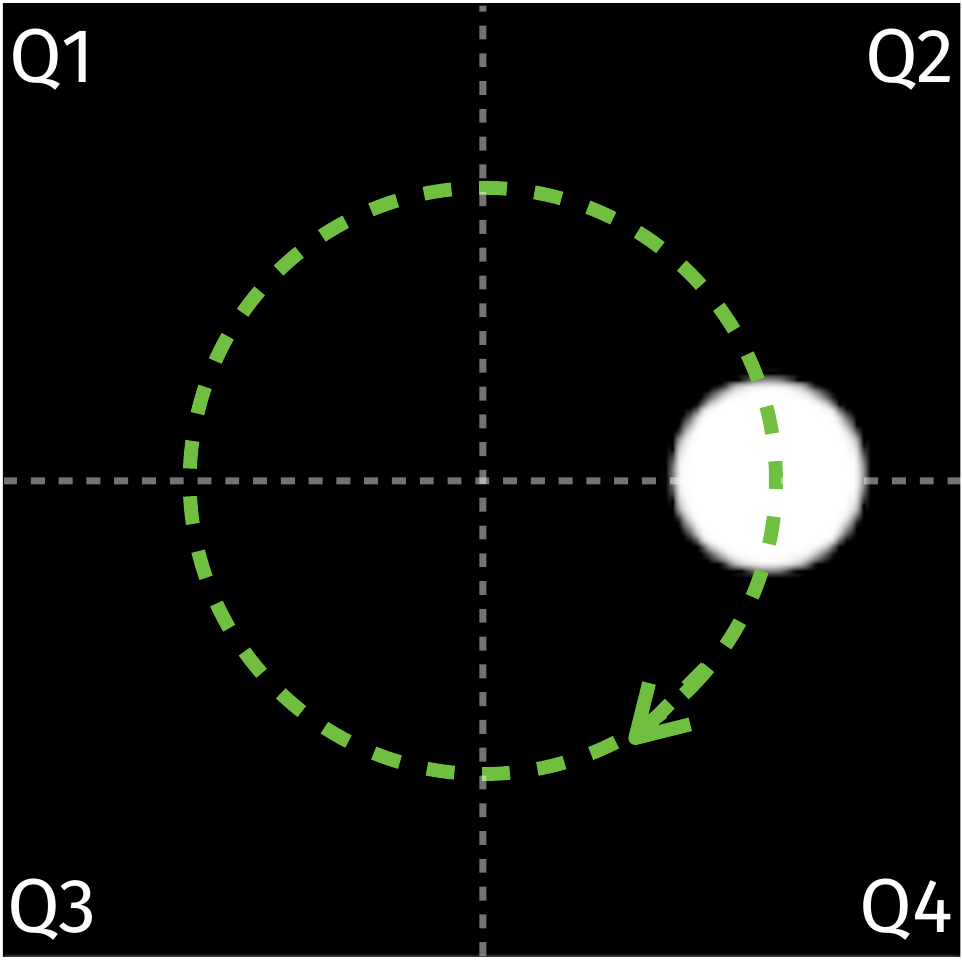}}\\
\includegraphics[width=0.7\textwidth]{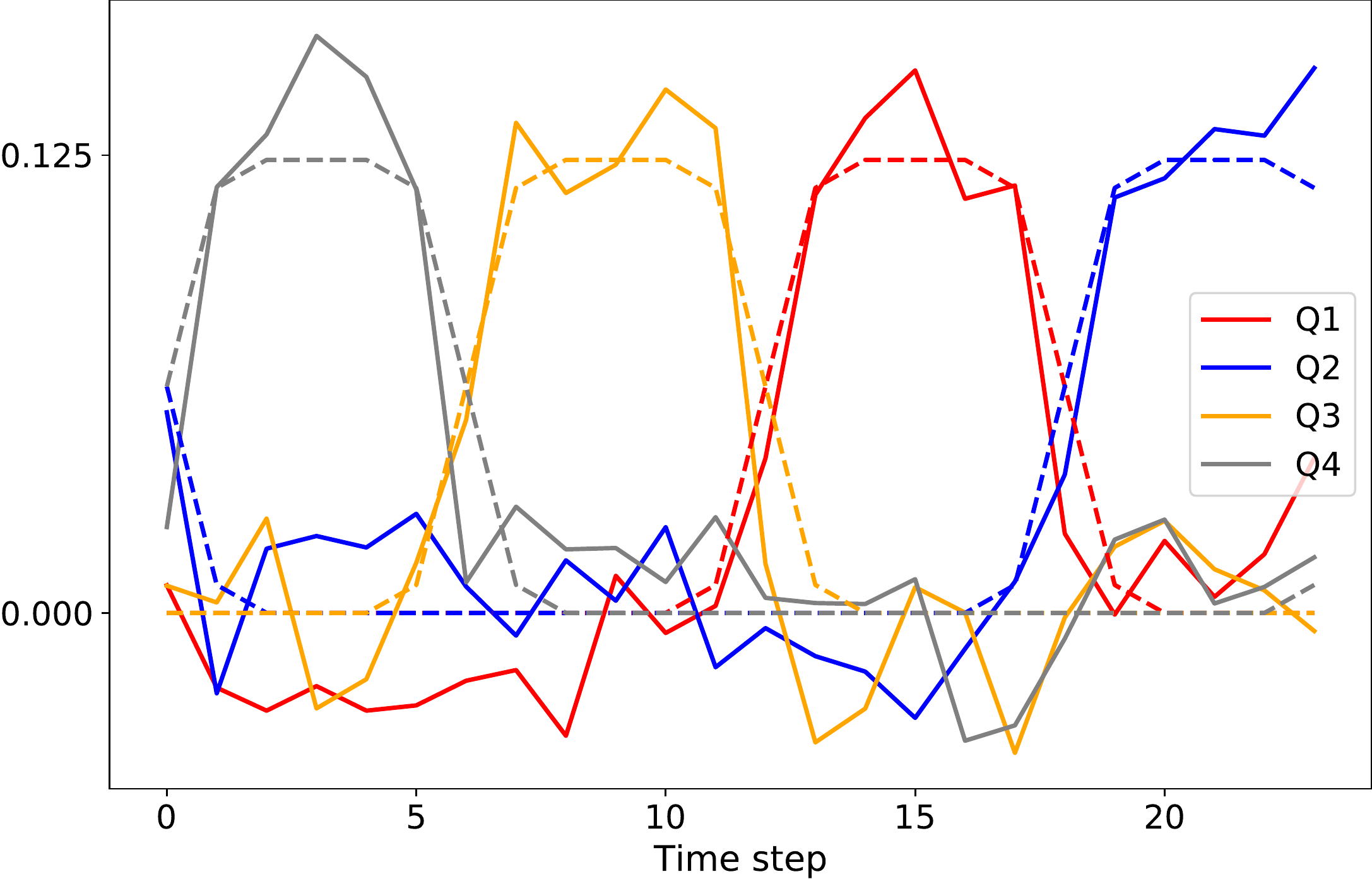}
\caption{(Top-left) Estimation of $\delta$ in~\eqref{eq:drop-spe}. (Top-right) Synthetic dynamic image sequence; a rotating white disk on a black background. (Bottom) The estimation of the quadrant occupancy functions per quadrant. 
}
\label{fig:delta-and-others}
\end{figure}

This section contains three experiments demonstrating how one can directly carry out signal estimation in the range of the DROP operator. Our first experiment aims at measuring the approximation error ($\delta$) in~\eqref{eq:drop-spe}. For this purpose, we took two unit vectors $\bs x, \bs u \in \bb R^n$ picked uniformly at random on the sphere $\bb S^{n-1}$ with the additional constraint that $\bs u$ must be orthogonal to $\bs x$. In this case, the term $\scp{\bs u}{\bs x}^2$ vanishes in~\eqref{eq:drop-spe} and one can estimate $\delta$ by Monte Carlo simulations over several trials of $\frac{\kappa}{m} \scp{\sign(\drop(\bs u))}{\drop(\bs x)}$. We show in Fig.~\ref{fig:delta-and-others}(top-left) the estimated $\delta$ over $100$ trials (average and maximum value) for $n=1000$ and  varying ratio $m/n \in [0.1,10]$. This estimation confirms that $\delta$ decays as $O(m^{-1/2})$ when $m$ increases. Moreover, on average, $\delta \simeq 0.1$ as soon as $m/n > 0.2$.         

As a second experiment, we test the possibility to perform localised detection in a simple synthetic video; a sequence of 24 images $\{\bs x_t\}_{t=0}^{23}$ of size $128\times 128$ ($n=16\,384$) representing a white disk rotating on a black (zero) background (see Fig.~\ref{fig:delta-and-others}(top-right)). Our objective is to detect the passage of the disk in each of the four image quadrants by only processing the DROP measurements $\{\drop(\bs x_t)\}_{t=0}^{23} \subset \bb R^m$. We have thus created 4 normalised patterns $\{\bs u_j\}_{j=1}^4$, with $\bs u_j$ being constant in the $j$-th quadrant and 0 outside, for each $1\leq j\leq 4$. We show in Fig.~\ref{fig:delta-and-others}(bottom) the time evolution of $q^{\rm est}_{j}(t) := \frac{\kappa}{m} \scp{\sign(\drop(\bs u_j))}{\drop(\bs x_t)}$ (continuous curves) for $m/n=0.5$ in comparison with the \emph{quadrant occupancy functions} $q_{j}(t) := \scp{\bs u_j}{\bs x_t}^2$ (dashed lines). The colour coding is given in the legend. We observe that each curves $q^{\rm est}_{j}(t)$ provides a fair estimation of $q_{j}(t)$. Moreover, the amplitude of each curve $q^{\rm est}_{j}(t)$ reduces when the disk is not in the associated quadrant. A rough detection of the quadrant occupancy could thus be established by appropriately thresholding each time signal $q^{\rm est}_{j}(t)$.

As final illustration, we perform a toy example (naive) classification of images taken in the labelled MNIST dataset $\cl X := \{(\bs x_k, t_k)\}_{k=1}^{N=70\,000} \subset \bb R^{28\times 28} \times \{0,\ldots,9\}$~\cite{mnist} (with normalisation $\|\bs x_k\|=1$). We want to compare the classification performed in the direct, pixel domain, to the one operated in the sketched domain. We start by randomly splitting $\cl X$ into a training $\cl X_{\rm tr}$ and a test set $\cl X_{\rm te}$ according to a split of $60\,000$ and $10\,000$ images, respectively. From $\cl X_{\rm tr}$ we compute the 10 centroids $\{\bs c_j\}_{j=0}^9 \subset \bb R^{28 \times 28}$ of each class of digits and we define 10 vectors $\bs u_j = \bs c_j /\|\bs c_j\|$. In the direct domain, the estimated label of an image $\bs x_k$ is then defined as $\hat{t}_k := \arg\max_j \scp{\bs u_j}{\bs x_k}^2$. In the sketched domain, the label estimate of a test image $\bs x_k$ is computed as $\hat{t}^{\rm sk}_k := \arg\max_j \frac{\kappa}{m} \scp{\sign(\drop(\bs u_j))}{\drop(\bs x_k)}$, which should be close to $\hat{t}_k$ according to Prop.~\ref{prop:drop-spe}. We reach the following average testing accuracy (100 trials, $\pm$ standard deviation in $\%$) for both the direct and the sketched classifications:
\begin{center}
\begin{tabular}{ |c|c|c|c|c|c| } 
 \hline
  & Direct & $m=200$ 
  & $m=400$ 
  & $m=800$ & $m=1600$\\ \hline
 Accuracy  $[\%]$ & $81.2$ & $69.3 \pm 1.9$
 & $75.2\pm 1.25$ 
 & $78.6\pm 0.91$&$80.4 \pm 0.79$\\ 
 \hline
\end{tabular}
\end{center}
These results confirm that, up to some distortion, one can apply this naive classification procedure directly in the DROP measurements of the dataset. As predicted by our analysis, this approximation improves when $m$ increases. 

To prospectively question the possibility to reach better classification accuracy in the DROP domain, we consider the dataset $\cl X^{\drop}= \{(\drop(\bs x_k), t_k)\}_{k=1}^{N=70\,000}$,  and compute the centroids $\{\bs c^{\drop}_j\}_{j=0}^9 \subset \bb R^m$ in the training split of $\cl X^{\drop}$. We then compare two possible classification methods. The first proceeds as a direct classification in the DROP domain, \ie we compute the label estimate $\hat{t}^{\drop}_k := \arg\max_j \scp{\bs c^{\drop}_j}{\drop(\bs x_k)}$. The second approach assumes that, for each centroid, there exists a vector $\bs v_j$ such that $\bs c^{\drop}_j \approx \drop(\bs v_j)$. Under this assumption, we rather estimate our label by applying a \emph{sign} operation to each $\bs c^{\drop}_j$---which is more aligned to Prop.~\ref{prop:drop-spe}--- and we compute $\hat{t}^{\drop, \sign}_k := \arg\max_j \scp{\sign(\bs c^{\drop}_j)}{\drop(\bs x_k)}$. We reach the following average testing accuracy for the two approaches (for $m=800$, 500 trials, $\pm$ standard deviation in \%):
\begin{center}
\begin{tabular}{ |c|c|c|c|c|c| } 
 \hline
  & Direct & Estimate $\hat{t}^{\rm sk}_k$ & Estimate $\hat{t}^{\cl B}_k$ & Estimate $\hat{t}^{\cl B, \sign}_k$\\ \hline
 Accuracy $[\%]$ & $81.2$ & $78.5 \pm 0.7$ & $63.4 \pm 3.3$  & $84.2 \pm 0.6$ \\ 
 \hline
\end{tabular}
\end{center}
Keeping only the sign of the estimated centroids provides better testing accuracy than the unsigned method. In fact, the signed approach outperforms the testing accuracy of the direct approach by about 3\%. 

\section{Conclusion and perspectives}

Our developments have shown that signal estimation is possible by directly processing the quadratic measurements of a signal. We achieved this by showing that such ROP measurements satisfy the sign product embedding with high probability provided $m$ is sufficiently large compared to the signal space dimension. Our theoretical results were backed up by several synthetic signal estimation and classification experiments. In future works, we plan to reproduce this experiment on an actual OPU to further reduce the computational cost of the procedure, by also leveraging the one-bit nature of the signed sketches. 

\appendix

\section{Proofs}
\label{sec:proofs}

In this appendix, we start by providing a proof of the isotropy of the DROP operator $\drop$. A more general proof is provided in~\cite[Lemma 4]{chen} (where $\drop$ is extended to the mapping of matrices of any rank and for sub-gaussian random vectors $\bs a_i$), but we find useful to provide below a short alternate proof for the sake of self-containedness. 

\begin{proposition}[DROP isotropy]
\label{prop:unbiased-drop}
Given $\bs x \in \bb R^n$, we have
\begin{equation}
\label{eq:unbiased-drop}
\ts\frac{1}{m} \bb E \|\drop(\bs x)\|^2 = 4\|\bs x\|^4.    
\end{equation}    
\end{proposition}
\begin{proof}
By rotational symmetry of the Gaussian distribution and by homogeneity of~\eqref{eq:unbiased-drop} in $\|\bs x\|^4$, it is enough to prove Prop.~\ref{prop:unbiased-drop} for $\bs x = \bs e_1 := (1,0,\ldots,0)^\top$. In this case, for two independent random vectors $\bs a, \bs b \sim \cl N(\bs 0, \bs I_n)$, we easily show that $\frac{1}{m}\bb E\|\drop(\bs x)\|^2 = \bb E \big[(\bs a^\top \bs x)^2-(\bs b^\top \bs x)^2\big]^2 = \bb E \big[a_1^2- b_1^2]^2  2 (\bb E a_1^4 - \bb E a_1^2 b_1^2) = 4 = 4 \|\bs x\|^4$.
\end{proof}

We now prove the central result of this work, Prop.~\ref{prop:drop-spe}, which is reproduced below for convenience.

\begin{nnproposition}[\ref{prop:drop-spe}]
Given a fixed unit vector $\bs u \in \bb R^n$, $\kappa = \pi/4$, and a distortion $0<\delta <1$, provided that 
\begin{equation}
\ts m \geq C \delta^{-2} k \log(\frac{n}{k\delta}), \tag{\ref{eq:sample-complex-drop-spe}}   
\end{equation}
then, with probability exceeding $1 - C \exp(-c \delta ^2 m)$, for all $k$-sparse signals $\bs x \in \Sigma_k := \{\bs v \in \bb R^n: |{\rm supp}(\bs v)|\leq k\}$, $\cl B$ respects the SPE over $\Sigma_k$, \ie 
\begin{equation}
\ts \Big|\frac{\kappa}{m} \scp{\sign(\drop(\bs u))}{\drop(\bs x)} - \scp{\bs u}{\bs x}^2 \Big| \leq \delta \|\bs x\|^2, \tag{\ref{eq:drop-spe}}   
\end{equation}    
with $\sign$ the sign operator applied componentwise on vectors. 
\end{nnproposition}

\begin{proof}

Let us first show that the expectation of $\frac{\kappa}{m} \scp{\sign(\drop(\bs u))}{\drop(\bs x)}$ is actually equal to $\scp{\bs u}{\bs x}^2$. 

By rotational symmetry of the Gaussian distribution and homogeneity of this expectation in $\|\bs x\|^2$, it is enough to prove it for $\bs u=\bs e_1$ and $\bs x = c \bs e_1 + s \bs e_2$, with $c = \cos\theta$, $s=\sin\theta$, and $\theta$ the angle between $\bs x$ and $\bs u$. 

In this case, for two independent random vectors $\bs a, \bs b \sim \cl N(\bs 0, \bs I_n)$, we easily show that \begin{align*}
\ts \frac{1}{m} \bb E \scp{\sign(\drop(\bs u))}{\drop(\bs x)}&= \bb E \sign(a_1^2 - b_1^2)[(c a_1 + s a_2)^2 - (c b_1 + s b_2)^2]\\
&= \bb E \sign(a_1^2 - b_1^2)[c^2(a_1^2 - b_1^2)+s^2(a_2^2 - b_2^2)+2cs (a_1 a_2 - b_1b_2)]\\
&= c^2 \bb E \sign(a_1^2 - b_1^2)(a_1^2 - b_1^2) = c^2 \bb E |a_1 - b_1||a_1 + b_1|.    
\end{align*}
Since $a_1 - b_1 \sim \cl N(0,2)$ and $a_1 + b_1 \sim \cl N(0,2)$ are decorrelated, and thus independent, we have $\bb E |a_1 - b_1|||a_1 + b_1| = 2 \bb E(2^{-1/2} |a_1 - b_1|)\, \bb E(2^{-1/2} |a_1 + b_1|) = \frac{4}{\pi} = 1/\kappa$, which proves the claim since $c^2 = \cos^2 \theta = \scp{\bs u}{\bs x}^2$. 

\medskip
Second, we study the concentration of $\frac{\kappa}{m} \scp{\sign(\drop(\bs u))}{\drop(\bs x)}$ around its mean for a fixed unit vector $\bs u$ and still assuming that $\|\bs x\|^2=1$ by homogeneity of~\eqref{eq:drop-spe}. One can write 
$$
S := \kappa \scp{\sign(\drop(\bs u))}{\drop(\bs x)} - m \scp{\bs u}{\bs x}^2,
$$ 
as the sum $S = \sum_{i=1}^m Z_i$ with 
$$
Z_i := s_i(\bs u)[(\bs a_{2i}^\top \bs x)^2 - (\bs a_{2i+1}^\top \bs x)^2],
$$ 
and $s_i(\bs u) = \sign[(\bs a_{2i}^\top \bs u)^2 - (\bs a_{2i+1}^\top \bs u)^2]$. The random variables $Z_i$ are all \iid and sub-exponential since their sub-exponential norm is bounded as 
$$
\|Z_i\|_{\psi_1} = \|(\bs a_{2i}^\top \bs x)^2 - (\bs a_{2i+1}^\top \bs x)^2\|_{\psi_1} \leq 2 \|(\bs a_{2i}^\top \bs x)^2\|_{\psi_1} \leq 4 \|\bs a_{2i}^\top \bs x\|^2_{\psi_2} \leq C,$$
see~\cite[Def. 5.13 \& Lemma 5.14]{versh12}. 

Therefore, from~\cite[Cor. 5.17]{versh12}, for any $\delta\geq 0$, we get the concentration result
$\bb P\big[ |S| \geq \delta m \big] \leq 2 \exp(-c \min(\delta^2, \delta) m)$, or 
\begin{equation}
\label{eq:conc-pre-spe}
\ts \bb P\big[ \big|\frac{\kappa}{m} \scp{\sign(\drop(\bs u))}{\drop(\bs x)} - \scp{\bs u}{\bs x}^2 \big| \geq \delta \big] \leq 2 \exp(-c \min(\delta^2, \delta) m). 
\end{equation}

\medskip
Finally, we extend the previous concentration result for a fixed unit vector $\bs u$ and for all vectors of $\Sigma_k$. By the same homogeneity argument, it is enough to prove the final claim of Prop.~\ref{prop:drop-spe} for all unit $k$-sparse signals $\bs x$ of $\Sigma_k \cap \bb S^{n-1}$. 

Let us observe that, for all $1 \leq i \leq m$,
$$
\ts \drop(\bs x) = \bs d_i^\top \bs X \bs s_i,\quad\text{with}\ \bs d_i := \bs a_{2i} - \bs a_{2i+1}, \bs s_i := \bs a_{2i} + \bs a_{2i+1},
$$
and $\bs X = \bs x \bs x^\top$. Moreover, the lifted matrix $\bs X$ has rank one, unit Frobenius norm and is sparse along its rows and columns, or \emph{bi-sparse}~\cite{fou20}. 

Therefore, given a radius $\epsilon >0$ to be fixed momentarily, there exists a $\epsilon$-covering $\cl G_\epsilon \subset \cl G$ of the lifted set $\cl G := \{\bs x \bs x^\top : \bs x \in \Sigma_k \cap \bb S^{n-1}\}$ such that for any $\bs X \in \cl G$ there is a $\bs X' \in \cl G_\epsilon$ with same row and column supports and $\|\bs X - \bs X\|\leq \epsilon$. Moreover, its cardinality $|\cl G_\epsilon|$ does not exceed $\binom{n}{k}(9/\epsilon)^{2k +1}$~\cite[Lemma 3.1]{canplan11}, which is crudely bounded by $(\frac{c n}{\epsilon k})^{2k+1}$. 

Fixing $\bar{\bs u} := \sign(\drop(\bs u))$, since 
$$
\ts \frac{\kappa}{m} \scp{\bar{\bs u}}{\drop(\bs x)} - \scp{\bs u}{\bs x}^2
= \frac{\kappa}{m} \scp{\bar{\bs u}}{\bs d_i^\top \bs X \bs s_i} - \bs u^\top \bs X \bs u,
$$
by a standard union bound argument applied to~\eqref{eq:conc-pre-spe}, we thus see that, with probability exceeding $1 - 2 |\cl G_\epsilon| \exp(-c \min(\delta^2, \delta) m)$, we have for all $\bs X' \in \cl G_\epsilon$ 
$$
\ts |\frac{\kappa}{m} \scp{\bar{\bs u}}{\bs d_i^\top \bs X' \bs s_i} - \bs u^\top \bs X' \bs u| \leq \delta.
$$
Therefore, provided $m \geq C \min(\delta,\delta^2)^{-1} k \log(\frac{n}{k\epsilon})$, this last inequality holds with probability greater than  $1 - 2 \exp(-c \min(\delta^2, \delta) m)$.

The rest of the proof assumes that this event holds. Let us define the $n \times n$ matrix 
$$
\ts \bs Q := \frac{\kappa}{m} \sum_{i=1}^m \bar{u}_i \bs d_i \bs s_i^\top - \bs u \bs u^\top.
$$
We then have for all $\bs X \in \cl G$, 
$$
\ts \frac{\kappa}{m} \scp{\bar{\bs u}}{\bs d_i^\top \bs X \bs s_i} - \bs u^\top \bs X \bs u = \scp{\bs Q}{\bs X}.
$$
Moreover, defining the maximal radius $\rho := \sup_{\bs X \in \cl G} |\scp{\bs Q}{\bs X}|$ and following a similar argument to~\cite[Sec. 3.2]{canplan11}, given an arbitrary $\bs X \in \cl G$, and for $\bs X' \in \cl G_\epsilon$ such that $\|\bs X - \bs X'\|\leq \epsilon$, we find
$$
|\scp{\bs Q}{\bs X}| \leq |\scp{\bs Q}{\bs X'}| + |\scp{\bs Q}{\bs Y}| \|\bs X - \bs X'\| \leq \delta + |\scp{\bs Q}{\bs Y}| \epsilon,
$$
with $\bs Y = (\bs X - \bs X')/\|\bs X - \bs X'\|$. Since the symmetric matrix $\bs Y$ has rank two, we can decompose it as $\bs Y = \bs Y_1 + \bs Y_2$ with $\hat{\bs Y}_i := \bs Y_i / \|\bs Y_i\| \in \cl G$ and $\scp{\bs Y_1}{\bs Y_2} = 0$. Therefore, $|\scp{\bs Q}{\bs Y}| \leq |\scp{\bs Q}{\hat{\bs Y}_1}| \|\bs Y_1\| + |\scp{\bs Q}{\hat{\bs Y}_2}| \|\bs Y_2\| \leq 2 \rho$ and     
$|\scp{\bs Q}{\bs X}| \leq \delta + 2 \rho \epsilon$. Taking the suppremum over $\bs X$, this means that $\rho \leq \delta + 2 \rho \epsilon$, or $\rho \leq \delta/(1-2\epsilon)$. Picking for instance $\epsilon = \delta /4$ and $\delta < 1$, gives $|\scp{\bs Q}{\bs X}| \leq \delta + \delta^2 \leq 2 \delta$, and a final rescaling of $\delta$ concludes the proof.     
\end{proof}

The following corollary extends Prop.~\ref{eq:drop-spe} to the case where $\bs u$ belongs to a finite set of unit vectors.
\begin{corollary}
\label{cor:drop-spe-set}
Given a set of $S$ fixed unit vector $\{\bs u_s\}_{s=1}^S$, $\kappa = \pi/4$, and a distortion $0<\delta <1$, provided that $m \geq C \delta^{-2} (k \log(\frac{n}{k\delta}) + \log S)$, we have, with probability exceeding $1 - C \exp(-c \delta ^2 m)$, for all $k$-sparse signals $\bs x \in \Sigma_k := \{\bs v \in \bb R^n: |{\rm supp}(\bs v)|\leq k\}$ and all $s\in [S]$,
\begin{equation}
\label{eq:drop-spe-set}
\ts \Big|\frac{\kappa}{m} \scp{\sign(\drop(\bs u_s))}{\drop(\bs x)} - \scp{\bs u_s}{\bs x}^2 \Big| \leq \delta \|\bs x\|^2.    
\end{equation}    
\end{corollary}
The proof of this corollary consists in bounding the failure of~\eqref{eq:drop-spe-set} using Prop.~\ref{prop:drop-spe} and a simple union bound argument.

\end{document}